\documentclass[12pt,a4paper,twoside]{article}

\usepackage{amsfonts,amsmath,amssymb,amsthm,latexsym}
\usepackage{times,mathptm}
\usepackage[all]{xy}
\usepackage{handelingen}

\newtheorem{proposition}{Proposition}[section]

\theoremstyle{definition}
\newtheorem{definition}[proposition]{Definition}

\theoremstyle{remark}

\newcommand{\selabel}[1]{\label{se:#1}}

\renewcommand{\thefootnote}{\fnsymbol{footnote}}
\title{\MakeUppercase{Graded structure and Hopf structures in parabosonic algebra. An
alternative approach  to bosonisation}\footnote{Talk presented at
the International Conference: ``\textit{New techniques in Hopf
algebras and graded ring theory}",  Brussels, September 19-23,
2006}}
\author{K. Kanakoglou$^{1}$ and C. Daskaloyannis$^{2}$}
\date{}
\renewcommand{\date}{\vspace{-5mm}}
\begin{document}
\maketitle \vspace*{-3mm}\relax
\renewcommand{\thefootnote}{\arabic{footnote}}

\noindent \textit{\small $^1$ Department of Physics, Aristotle University of Thessaloniki, Thessaloniki 54124, GREECE\\
e-mail: kanakoglou@hotmail.com }\\
\noindent \textit{\small $^2$ Department of Mathematics, Aristotle
University of Thessaloniki, Thessaloniki 54124, GREECE\\
 e-mail: daskalo@math.auth.gr}

\begin{abstract}
\noindent Parabosonic algebra in infinite degrees of freedom is
presented as a generalization of the bosonic algebra, from the
viewpoints of both physics and mathematics. The notion of
super-Hopf algebra is shortly discussed and the super-Hopf
algebraic structure of the parabosonic algebra is established
(without appealing to its Lie superalgebraic structure). Two
possible variants of the parabosonic algebra are presented and
their (ordinary) Hopf algebraic structure is estabished: The first
is produced by ``bosonising" the original super-Hopf algebra,
while the second is constructed via a slightly different path.
\end{abstract}

\section*{Introduction}
Parabosonic algebras have a long history both in theoretical and
mathematical physics. Although, formally introduced in the fifties
by Green \cite{Green}, in the context of second quantization,
their history traces back to the fundamental conceptual problems
of quantum mechanics; in particular to Wigner's approach to first
quantization \cite{Wi}. We begin by outlining this story.

In classical physics, all information describing the dynamics of a
given physical system is ``encoded" in it's Hamiltonian $H(p_{i},
q_{i})$, $i=1,\ldots,n$ which is a function of the real variables
$p_{i}, q_{i}$. These are usually called ``canonical variables".
Having determined the Hamiltonian of the system, the dynamics is
extracted through the well-known Hamilton equations:
\begin{equation} \label{Hamilton}
\frac{dq_{i}}{dt} = \frac{\partial H}{\partial p_{i}} \ \ \ \ \ \
\ \ \ \ \ \frac{dp_{i}}{dt} = - \frac{\partial H}{\partial q_{i}}
\end{equation}
The passage from the classical description to the quantum
description, within the framework of the first quantization,
consists of the following procedure: The functional dependence of
the Hamiltonian on the canonical variables is -roughly- retained
but the canonical variables are no more real variables. Instead
they become elements of a unital associative non-commutative
algebra, described in terms of the generators $p_{i}, q_{i}, I$,
$i=1,\ldots,n$ and relations:
\begin{equation} \label{CCR}
[q_{i}, p_{j}] = i \hbar \delta_{ij} I \ \ \ \ \ \ \ \ \ \ [q_{i},
q_{j}] = [p_{i}, p_{j}] = 0
\end{equation}
$I$ is of course the unity of the algebra and $[x, y]$ stands for
$xy-yx$. The states of the system are no more described as
functions $p_{i}(t), q_{i}(t)$ which are solutions of
\eqref{Hamilton} but rather as vectors of a Hilbert space, where
the elements of the above mentioned algebra act. The dynamics is
now determined by the Heisenberg equations of motion:
\begin{equation} \label{Heisenberg}
i \hbar \frac{dq_{i}}{dt} = [q_{i}, H] \ \ \ \ \ \ \ \ \ i \hbar
\frac{dp_{i}}{dt} = [p_{i}, H]
\end{equation}
We of course describe what is known as the Heisenberg picture of
quantum mechanics. Relations \eqref{CCR} are known in the physical
literature as the Heisenberg algebra, or the Heisenberg-Weyl
algebra or more commonly as the Canonical Commutation Relations
often abbreviated as CCR. Their central importance for the
quantization procedure formerly described, lies in the fact that
if one accepts \footnote{of course we do not consider arbitrary
Hamiltonians but functions of the form $H =
\sum_{i=1}^{n}p_{i}^{2} + V(q_{1},\ldots,q_{n})$ which however are
general enough for simple physical systems} the algebraic
relations \eqref{CCR} together with the quantum dynamical
equations \eqref{Heisenberg} then it is an easy matter (see
\cite{Ehrenf}) to extract the classical Hamiltonian equations of
motion \eqref{Hamilton} while on the other hand the acceptance of
the classical equations \eqref{Hamilton} together with \eqref{CCR}
reproduces the quantum dynamics exactly as described by
\eqref{Heisenberg}. In this way the CCR emerge as a fundamental
link between the classical and the quantum description of the
dynamics.

 For  technical reasons it is common to use -instead of
the variables $p_{i}, q_{i}$- the linear combinations:
\begin{equation} \label{creation-destruction op}
b_{j}^{+} = \frac{1}{\sqrt{2}}(q_{j} - ip_{j}) \ \ \ \ \ \ \ \
b_{j}^{-} = \frac{1}{\sqrt{2}}(q_{j} + ip_{j})
\end{equation}
for $j=1,\ldots,n$ in terms of which \eqref{CCR} become (we have
set $\hbar = 1$):
\begin{equation} \label{CCRbose}
[b_{i}^{-}, b_{j}^{+}] =  \delta_{ij} I \ \ \ \ \ \ \ \ \ \
[b_{i}^{-}, b_{j}^{-}] = [b_{i}^{+}, b_{j}^{+}] = 0
\end{equation}
for $i,j=1,\ldots,n$. These latter relations are usually called
the bosonic algebra (of n bosons), and in they case of the
infinite degrees of freedom $i,j = 1, 2, \ldots \ $ they become
the starting point of the free field theory (i.e.: second
quantisation).

The above mentioned approach to the first quantisation is actually
the path followed by the founders of quantum mechanics such as
Dirac, Born, Heisenberg, Schr\"{o}edinger and others. Although it
is not our aim to provide systematic references on this
fascinating story, many of the original papers which paved the way
can be found in \cite{VanDerWar}.

In 1950 E.P. Wigner in a two page publication \cite{Wi}, noticed
that what the above approach implies is that the CCR \eqref{CCR}
are sufficient conditions -but not necessary- for the equivalence
between the classical Hamiltonian equations \eqref{Hamilton} and
the Heisenberg quantum dynamical equations \eqref{Heisenberg}. In
a kind of reversing the problem, Wigner posed the question of
looking for necessary conditions for the simultaneous fulfillment
of \eqref{Hamilton} and \eqref{Heisenberg}. He stated an infinite
set of solutions for the above mentioned problem (although not
claiming to have found the general solution). It is worth noting
that CCR were included as one special case among Wigner's infinite
solutions.

A few years latter in 1953, Green in his celebrated paper
\cite{Green} introduced the parabosonic algebra (in possibly
infinite degrees of freedom), by means of generators and
relations:
\begin{equation} \label{CCRparabose}
\begin{array}{c}
  \big[ B_{m}^{-}, \{ B_{k}^{+}, B_{l}^{-} \} \big] = 2\delta_{km}B_{l}^{-}  \\
           \\
  \big[ B_{m}^{-}, \{ B_{k}^{-}, B_{l}^{-} \} \big]= 0 \\
           \\
  \big[ B_{m}^{+}, \{ B_{k}^{-}, B_{l}^{-} \} \big] = - 2\delta_{lm}B_{k}^{-} - 2\delta_{km}B_{l}^{-}\\
\end{array}
\end{equation}
$k,l,m = 1, 2, \ldots$ and $\{x, y \}$ stands for $xy+yx$. Green
was primarily interested in field theoretic implications of the
above mentioned algebra, in the sense that he considered it as an
alternative starting point for the second quantisation problem,
generalizing \eqref{CCRbose}. However, despite his original
motivation he was the first to realize -see also \cite{OhKa}- that
Wigner's infinite solutions were nothing else but inequivalent
irreducible representations of the parabosonic algebra
\eqref{CCRparabose}. (See also the discussion in \cite{Pal1}).

In what follows, all vector spaces and algebras and all tensor
products will be considered over the field of complex numbers.

\section{Preliminaries: Bosons and Parabosons as superalgebras} \selabel{1}

The parabosonic algebra, was originally defined in terms of
generators and relations by Green \cite{Green} and
Greenberg-Messiah \cite{GreeMe}. We begin with restating their
definition: \\
 Let us consider the vector space $V_{X}$ freely
generated by the elements: $X_{i}^{+}, X_{j}^{-}$, $i,j=1, 2,
...$. Let $T(V_{X})$ denote the tensor algebra of $V_{X}$.
$T(V_{X})$ is -up to isomorphism- the free algebra generated by
the elements of the basis. In $T(V_{X})$ we consider the two-sided
ideals $I_{P_{B}}$, $I_{B}$, generated by the following elements:
\begin{equation} \label{eq:pbdef}
 \big[ \{ X_{i}^{\xi},  X_{j}^{\eta}\}, X_{k}^{\epsilon}  \big] -
 (\epsilon - \eta)\delta_{jk}X_{i}^{\xi} - (\epsilon - \xi)\delta_{ik}X_{j}^{\eta}
\end{equation}
and:
\begin{equation} \label{eq:bdef}
\begin{array}{ccccc}
 \big[ X_{i}^{-}, X_{j}^{+} \big] - \delta_{ij}I_{X} & , & \big[ X_{i}^{-}, X_{j}^{-} \big]
 & , & \big[ X_{i}^{+},  X_{j}^{+} \big]
   \\
\end{array}
\end{equation}
respectively, for all values of $\xi, \eta, \epsilon = \pm 1$ and
$i,j=1, 2, ... \ $. $ \ I_{X}$ is the unity of the tensor algebra.
We now have the following:
\begin{definition} \label{parabosonsbosons}
The parabosonic algebra in $P_{B}$  is the quotient algebra of the
tensor algebra $T(V_{X})$ of $V_{X}$ with the ideal $I_{P_{B}}$:
$$
P_{B} = T(V_{X}) / I_{P_{B}}
$$
The bosonic algebra  $B$ is the quotient algebra of the tensor
algebra $T(V_{X})$ with the ideal $I_{B}$:
$$
B = T(V_{X}) / I_{B}
$$
\end{definition}
We denote by $\pi_{P_{B}}: T(V_{X}) \rightarrow P_{B}$ and
$\pi_{B}: T(V_{X}) \rightarrow B$ respectively, the canonical
projections. The elements $X_{i}^{+}, X_{j}^{-}, I_{X}$, where
$i,j=1, 2, ... \ $ and $I_{X}$ is the unity of the tensor algebra,
are the generators of the tensor algebra $T(V_{X})$. The elements
$\pi_{P_{B}}(X_{i}^{+}), \pi_{P_{B}}(X_{j}^{-}),
\pi_{P_{B}}(I_{X}) \ $, $ \ i,j=1,...$ are a set of generators of
the parabosonic algebra $P_{B}$, and they will be denoted by
$B_{i}^{+}, B_{j}^{-}, I$ for $i,j=1, 2, ...$ respectively, from
now on. $\pi_{P_{B}}(I_{X}) = I$ is the unity of the parabosonic
algebra. On the other hand elements $\pi_{B}(X_{i}^{+}),
\pi_{B}(X_{j}^{-}), \pi_{B}(I_{X}) \ $, $ \ i,j=1, 2, ...$ are a
set of generators of the bosonic algebra $B$, and they will be
denoted by $b_{i}^{+}, b_{j}^{-}, I$ for $i,j=1, 2, ...$
respectively, from now on. $\pi_{B}(I_{X}) = I$ is the unity of
the bosonic algebra.

Based on the above definitions we prove now the following
proposition which clarifies the relationship between bosonic and
parabosonic algebras:
\begin{proposition} \label{parabosonstobosons}
The parabosonic algebra $P_{B}$ and the bosonic algebra $B$ are
both $\mathbb{Z}_{2}$-graded algebras with their generators
$B_{i}^{\pm}$ and $b_{i}^{\pm}$ respectively, $i,j=1, 2, ...$,
being odd elements. The bosonic algebra $B$ is a quotient algebra
of the parabosonic algebra $P_{B}$. The ``replacement" map $\phi:
P_{B} \rightarrow B$ defined by: $\phi(B_{i}^{\pm}) = b_{i}^{\pm}$
is a $\mathbb{Z}_{2}$-graded algebra epimorphism (i.e.: an even
algebra epimorphism).
\end{proposition}
\begin{proof}
It is obvious that the tensor algebra $T(V_{X})$ is a
$\mathbb{Z}_{2}$-graded algebra with the monomials being
homogeneous elements. If $x$ is an arbitrary monomial of the
tensor algebra, then $deg(x) = 0$, namely $x$ is an even element,
if it constitutes of an even number of factors (an even number of
generators of $T(V_{X})$) and $deg(x) = 1$, namely $x$ is an odd
element, if it constitutes of an odd number of factors (an odd
number of generators of $T(V_{X})$). The generators $X_{i}^{+},
X_{j}^{-} \ $, $ \ i,j=1,...,n$ are odd elements in the above
mentioned gradation.
 In view of the above description we can easily conclude that the
$\mathbb{Z}_{2}$-gradation of the tensor algebra is immediately
``transfered" to the algebras $P_{B}$ and $B$: Both ideals
$I_{P_{B}}$ and $I_{B}$ are homogeneous ideals of the tensor
algebra, since they are generated by homogeneous elements of
$T(V_{X})$. Consequently, the projection homomorphisms
$\pi_{P_{B}}$ and $\pi_{B}$ are homogeneous algebra maps of degree
zero, or we can equivalently say that they are even algebra
homomorphisms.
 We can straightforwardly check that the bosons
satisfy the paraboson relations, i.e:
$$
\begin{array}{c}
  \pi_{B}( \big[ \{ X_{i}^{\xi},  X_{j}^{\eta}\},
X_{k}^{\epsilon}  \big] -
 (\epsilon - \eta)\delta_{jk}X_{i}^{\xi} - (\epsilon -
 \xi)\delta_{ik}X_{j}^{\eta} ) = \\
              \\
 = \big[ \{ b_{i}^{\xi},  b_{j}^{\eta}\},
b_{k}^{\epsilon}  \big] -
 (\epsilon - \eta)\delta_{jk}b_{i}^{\xi} - (\epsilon -
 \xi)\delta_{ik}b_{j}^{\eta} = 0 \\
\end{array}
$$
which simply means that: $ker(\pi_{P_{B}}) \subseteq ker(\pi_{B})$
or equivalently: $I_{P_{B}} \subseteq I_{B}$. By the correspodence
theorem for rings, we get that the set $I_{B} / I_{P_{B}} =
\pi_{P_{B}}(I_{B})$ is an homogeneous ideal of the algebra
$P_{B}$, and applying the third isomorphism theorem for rings we
get:
\begin{equation} \label{thirdisomortheor}
P_{B} \Big/ (I_{B} / I_{P_{B}}) = (T(V_{X}) / I_{P_{B}}) \Big/
(I_{B} / I_{P_{B}}) \cong T(V_{X}) \Big/ I_{B} = B
\end{equation}
Thus we have shown that the bosonic algebra $B$ is a quotient
algebra of the parabosonic algebra $P_{B}$. The fact that
$I_{P_{B}} \subseteq I_{B}$ implies that $\pi_{B}$ is uniquely
extended to an even algebra homomorphism $\phi: P_{B} \rightarrow
B$, where $\phi$ is determined by it's values on the generators
$B_{i}^{\pm}$ of $P_{B}$, i.e.: $\phi(B_{i}^{\pm}) = b_{i}^{\pm}$.
Recalling now that: $ker\phi = I_{B} / I_{P_{B}} =
\pi_{P_{B}}(I_{B})$ and using equation \eqref{thirdisomortheor},
we get that: $P_{B} / ker\phi \cong B$ which completes the proof
that $\phi$ is an epimorphism of $\mathbb{Z}_{2}$-graded algebras
(or: an even epimorphism).
\end{proof}
Note that $ker\phi$ is exactly the ideal of $P_{B}$ generated by
the elements of the form: $ \ \big[ B_{i}^{-}, B_{j}^{+} \big] -
\delta_{ij}I \ $, $ \ \big[ B_{i}^{-}, B_{j}^{-} \big] \ $, $ \
\big[ B_{i}^{+},  B_{j}^{+} \big] \ $ for all values of $i,j=1, 2,
...$, and $I$ is the unity of the $P_{B}$ algebra.

 The rise of the
theory of quasitriangular Hopf algebras from the mid-80's
\cite{Dri} and thereafter and especially the study and abstraction
of their representations (see: \cite{Maj1, Maj2}, \cite{Mon} and
references therein), has provided us with a novel understanding
\footnote{it is worth noting, that some of these ideas already
appear in \cite{Stee}} of the notion and the properties of
$\mathbb{G}$-graded algebras,
where $\mathbb{G}$ is a finite abelian group:  \\
 Restricting ourselves to the simplest case
where $\mathbb{G} = \mathbb{Z}_{2}$, we recall that an algebra $A$
being a $\mathbb{Z}_{2}$-graded algebra (in the physics literature
the term superalgebra is also of widespread use) is equivalent to
saying that $A$ is a $\mathbb{CZ}_{2}$-module algebra, via the
$\mathbb{Z}_{2}$-action: $g \vartriangleright a = (-1)^{|a|}a \ $
(for $a$ homogeneous in $A$). What we actually mean is that $A$,
apart from being an algebra is also a $\mathbb{CZ}_{2}$-module and
at the same time it's structure maps (i.e.: the multiplication and
the unity map which embeds the field into the center of the
algebra) are $\mathbb{CZ}_{2}$-module maps which is nothing else
but homogeneous linear maps of degree $0$ (i.e.: even linear
maps). Note, that under the above action, any element of $A$
decomposes uniquely as: $a = \frac{a + (g \vartriangleright a)}{2}
+ \frac{a - (g \vartriangleright a)}{2}$. We can further summarize
the above description saying that $A$ is an algebra in the braided
monoidal category of $\mathbb{CZ}_{2}$-modules
${}_{\mathbb{CZ}_{2}}\mathcal{M}$. In this case the braiding is
induced by the non-trivial quasitriangular structure of the
$\mathbb{CZ}_{2}$ Hopf algebra i.e. by the non-trivial $R$-matrix:
\begin{equation} \label{eq:nontrivRmatrcz2}
R_{g} = \frac{1}{2}(1 \otimes 1 + 1 \otimes g + g \otimes 1 - g
\otimes g)
\end{equation}
In the above relation $1, g$ are the elements of the
$\mathbb{Z}_{2}$ group (written multiplicatively).

 We digress here for a moment, to recall
that (see \cite{Maj1, Maj2} or \cite{Mon}) if $(H,R_{H})$ is a
quasitriangular Hopf algebra, then the category of modules
${}_{H}\mathcal{M}$ is a braided monoidal category, where the
braiding is given by a natural family of isomorphisms $\Psi_{V,W}:
V \otimes W \cong W \otimes V$, given explicitly by:
\begin{equation} \label{eq:braid}
\Psi_{V,W}(v \otimes w) = \sum (R_{H}^{(2)} \vartriangleright w)
\otimes (R_{H}^{(1)} \vartriangleright v)
\end{equation}
for any $V,W \in obj({}_{H}\mathcal{M})$. By $v,w$ we denote any
 elements of $V,W$ respectively.   \\
Combining eq. \eqref{eq:nontrivRmatrcz2} and \eqref{eq:braid} we
immediately get the braiding in the
${}_{\mathbb{CZ}_{2}}\mathcal{M}$ category:
\begin{equation} \label{symmbraid}
\Psi_{V,W}(v \otimes w) = (-1)^{|v||w|} w \otimes v
\end{equation}
In the above relation $\ |.| \ $ denotes the degree of an
homogeneous element of either $V$ or $W$ (i.e.: $|x|=0$ if $x$ is
an even element and $|x|=1$ if $x$ is an odd element).
 This is obviously a symmetric braiding,
since $\Psi_{V,W} \circ \Psi_{W,V} = Id$, so we actually have a
symmetric monoidal category ${}_{\mathbb{CZ}_{2}}\mathcal{M}$,
rather than a truly braided one.

 The really important thing about
the existence of the braiding \eqref{symmbraid} is that it
provides us with an alternative way of forming tensor products of
$\mathbb{Z}_{2}$-graded algebras: If $A$ is a superalgebra with
multiplication $m: A \otimes A \rightarrow A$,  then the super
vector space $A \otimes A$ (with the obvious
$\mathbb{Z}_{2}$-gradation) equipped with the associative
multiplication
\begin{equation} \label{braidedtenspr}
(m \otimes m)(Id \otimes \Psi_{A,A} \otimes Id): A \otimes A
\otimes A \otimes A \longrightarrow A \otimes A
\end{equation}
 given by: $(a \otimes b)(c \otimes d) = (-1)^{|b||c|}ac \otimes
 bd$ ($b,c$ homogeneous in $A$), readily becomes a superalgebra (or equivalently an algebra in the
braided monoidal category of $\mathbb{CZ}_{2}$-modules
${}_{\mathbb{CZ}_{2}}\mathcal{M}$) which we will denote: $A
\underline{\otimes} A$ and call the braided tensor product algebra
from now on.

\section{Main results}\selabel{2}

\subsection{Parabosons as super-Hopf algebras}

The notion of $\mathbb{G}$-graded Hopf algebra, for $\mathbb{G}$ a
finite abelian group, is not a new one neither in physics nor in
mathematics. The idea appears already in the work of Milnor and
Moore \cite{MiMo}, where we actually have $\mathbb{Z}$-graded Hopf
algebras. On the other hand, universal enveloping algebras of Lie
superalgebras are widely used in physics and they are examples of
$\mathbb{Z}_{2}$-graded Hopf algebras (see for example \cite{Ko},
\cite{Scheu}). These structures are strongly resemblant of Hopf
algebras but they are not Hopf algebras at least in the ordinary
sense.

Restricting again to the simplest case where $\mathbb{G} =
\mathbb{Z}_{2}$ we briefly recall this idea: An algebra $A$ being
a $\mathbb{Z}_{2}$-graded Hopf algebra (or super-Hopf algebra)
means first of all that $A$ is a $\mathbb{Z}_{2}$-graded
associative algebra (or: superalgebra). We now consider the
braided tensor product algebra $A \underline{\otimes} A$. Then $A$
is equipped with a coproduct
$$
\underline{\Delta} : A \rightarrow A \underline{\otimes} A
$$
which is an superalgebra homomorphism from $A$ to the braided
tensor product algebra  $A \underline{\otimes} A$ :
$$
\underline{\Delta}(ab) = \sum (-1)^{|a_{2}||b_{1}|}a_{1}b_{1}
\otimes a_{2}b_{2} = \underline{\Delta}(a) \cdot
\underline{\Delta}(b)
$$
for any $a,b$ in $A$, with $\underline{\Delta}(a) = \sum a_{1}
\otimes a_{2}$, $\underline{\Delta}(b) = \sum b_{1} \otimes
b_{2}$, and $a_{2}$, $b_{1}$ homogeneous.

 Similarly, $A$ is equipped with an antipode $\underline{S} : A
\rightarrow A$ which is not an algebra anti-homomorphism (as in
ordinary Hopf algebras) but a  superalgebra anti-homomorphism (or:
``twisted" anti-homomorphism or: braided anti-homomorphism) in the
following sense (for any homogeneous $a,b \in A$):
$$
\underline{S}(ab) = (-1)^{|a||b|}\underline{S}(b)\underline{S}(a)
$$
The rest of the axioms which complete the super-Hopf algebraic
structure (i.e.: coassociativity, counity property, and
compatibility with the antipode) have the same formal description
as in ordinary Hopf algebras.

 Once again, the abstraction of the representation theory of
quasitriangular Hopf algebras provides us with a language in which
the above description becomes much more compact: We simply say
that $A$ is a Hopf algebra in the braided monoidal category of
$\mathbb{CZ}_{2}$-modules ${}_{\mathbb{CZ}_{2}}\mathcal{M}$ or: a
braided group where the braiding is given in equation
\eqref{symmbraid}. What we actually mean is that $A$ is
simultaneously an algebra, a coalgebra and a
$\mathbb{CZ}_{2}$-module, while all the structure maps of $A$
(multiplication, comultiplication, unity, counity and the
antipode) are also $\mathbb{CZ}_{2}$-module maps and at the same
time the comultiplication $\underline{\Delta} : A \rightarrow A
\underline{\otimes} A$ and the counit are algebra morphisms in the
category ${}_{\mathbb{CZ}_{2}}\mathcal{M}$ (see also \cite{Maj1,
Maj2} or \cite{Mon} for a more detailed description).
\\
We proceed now to the proof of the following proposition which
establishes the super-Hopf algebraic structure of the parabosonic
algebra $P_{B}$:
\begin{proposition}
The parabosonic algebra  equipped with the even linear maps
$\underline{\Delta}: P_{B} \rightarrow P_{B} \underline{\otimes}
P_{B} \ \ $, $\ \ \underline{S}: P_{B} \rightarrow P_{B} \ \ $, $\
\ \underline{\varepsilon}: P_{B} \rightarrow \mathbb{C} \ \ $,
determined by their values on the generators:
\begin{equation} \label{eq:HopfPB}
\begin{array}{ccccc}
  \underline{\Delta}(B_{i}^{\pm}) = 1 \otimes B_{i}^{\pm} + B_{i}^{\pm} \otimes 1 &
  & \underline{\varepsilon}(B_{i}^{\pm}) = 0  & & \underline{S}(B_{i}^{\pm}) = - B_{i}^{\pm} \\
\end{array}
\end{equation}
becomes a super-Hopf algebra.
\end{proposition}
\begin{proof}
Recall that by definition $P_{B} = T(V_{X}) / I_{P_{B}}$. Consider
the linear map: $\underline{\Delta}: V_{X} \rightarrow P_{B}
\underline{\otimes} P_{B}$ determined by it's values on the basis
elements specified by: $\underline{\Delta}(X_{i}^{\pm}) = I
\otimes B_{i}^{\pm} + B_{i}^{\pm} \otimes I$. By the universality
of the tensor algebra this map is uniquely extended to a
superalgebra homomorphism: $\underline{\Delta}: T(V_{X})
\rightarrow P_{B} \underline{\otimes} P_{B}$. Now we compute:
$$
\underline{\Delta}(\big[ \{ X_{i}^{\xi},  X_{j}^{\eta}\},
X_{k}^{\epsilon}  \big] -
 (\epsilon - \eta)\delta_{jk}X_{i}^{\xi} - (\epsilon -
 \xi)\delta_{ik}X_{j}^{\eta})= 0
$$
This means that $I_{P_{B}} \subseteq ker\underline{\Delta}$, which
in turn implies that $\underline{\Delta}$ is uniquely extended as
a superalgebra homomorphism: $\underline{\Delta}: P_{B}
\rightarrow P_{B} \underline{\otimes} P_{B}$, with values on the
generators determined by \eqref{eq:HopfPB}. Proceeding the same
way we construct the maps $\ \underline{\varepsilon} \ $, $\ \
\underline{S} \ $, as determined in \eqref{eq:HopfPB}.

Note here that in the case of the antipode $\underline{S}$ we need
the notion of the $\mathbb{Z}_{2}$-graded opposite algebra (or:
opposite superalgera) $P_{B}^{op}$, which is a superalgebra
defined as follows: $P_{B}^{op}$ has the same underlying super
vector space as $P_{B}$, but the multiplication is now defined as:
$a \cdot b = (-1)^{|a||b|}ba$, for all $a,b \in P_{B}$. (In the
right hand side, the product is of course the product of $P_{B}$).
We start by defining a linear map $\underline{S}: V_{X}
\rightarrow P_{B}^{op}$ by: $\underline{S}(X_{i}^{\pm}) =
-B_{i}^{\pm}$ which is (uniquely) extended to a superalgebra
homomorphism: $\underline{S}: T(V_{X}) \rightarrow P_{B}^{op}$.
The fact that $I_{P_{B}} \subseteq ker\underline{S}$ implies that
$\underline{S}$ is uniquely extended to a superalgebra
homomorphism $\underline{S}: P_{B} \rightarrow P_{B}^{op}$, thus
to a superalgebra anti-homomorphism: $\underline{S}: P_{B}
\rightarrow P_{B}$ with values on the generators determined by
\eqref{eq:HopfPB}.

Now it is sufficient to verify the rest of the super-Hopf algebra
axioms (coassociativity, counity and the compatibility condition
for the antipode) on the generators of $P_{B}$. This can be done
with straigthforward computations.
\end{proof}
Let us note here, that the above proposition generalizes a result
which -in the case of finite degrees of freedom- is a direct
consequence of the work in \cite{Pal}. In that work the
parabosonic algebra in $2n$ generators ($n$-paraboson algebra)
$P_{B}^{(n)}$ is shown to be isomorphic to the universal
enveloping algebra of the orthosymplectic Lie superalgebra:
$P_{B}^{(n)} \cong U(B(0,n))$. See also the discussion in
\cite{KaDa1, KaDa2}.

\subsection{Ordinary Hopf structures for parabosons}

\subsubsection{Review of the bosonisation technique}

A general scheme for ``transforming" a Hopf algebra $A$ in the
braided category ${}_{H}\mathcal{M}$ ($H$: some quasitriangular
Hopf algebra) into an ordinary one, namely the smash product Hopf
algebra: $A \star H$, such that the two algebras have equivalent
module categories, has been developed during '90 's. The original
reference is \cite{Maj1} (see also \cite{Maj2, Maj3}). The
technique is called bosonisation, the term coming from physics.
This technique uses ideas developed in \cite{Ra}, \cite{Mo}. It is
also presented and applied in \cite{Fi}, \cite{FiMon},
\cite{Andru}. We review the main points of the above method:

 In general, $A$ being a Hopf algebra in a category, means that
 $A$ apart from being an algebra and a coalgebra, is also an object of
 the category and at the same time it's structure maps are morphisms in the category.
 In particular, if $H$ is some quasitriangular Hopf algebra, $A$ being
a Hopf algebra in the braided monoidal category
${}_{H}\mathcal{M}$ , means that the $H$-module $A$ is an algebra
in ${}_{H}\mathcal{M}$ (or: $H$-module algebra) and a coalgebra in
${}_{H}\mathcal{M}$ (or: $H$-module coalgebra) and at the same
time $\Delta_{A}$ and $\varepsilon_{A}$ are algebra morphisms in
the category ${}_{H}\mathcal{M}$. (For more details
on the above definitions one may consult for example \cite{Mon}). \\
Since $A$ is an $H$-module algebra we can form the cross product
algebra $A \rtimes H$ (also called: smash product algebra) which
as a k-vector space is $A \otimes H$ (i.e. we write: $a \rtimes h
\equiv a \otimes h$ for every $a \in A$, $h \in H$), with
multiplication given by:
\begin{equation} \label{eq:crosspralg}
(b \otimes h)(c \otimes g) = \sum b(h_{1} \vartriangleright c)
\otimes h_{2}g
\end{equation}
$\forall$ $b,c \in A$ and $h,g \in H$, and the usual tensor
product unit. \\
On the other hand $A$ is a (left) $H$-module coalgebra with $H$:
quasitriangular through the $R$-matrix: $R_{H} = \sum R_{H}^{(1)}
\otimes R_{H}^{(2)}$.
 Quasitriangularity
``switches" the (left) action of $H$ on $A$ into a (left) coaction
$\rho: A \rightarrow H \otimes A$ through:
\begin{equation} \label{eq:act-coact}
\rho(a) = \sum R_{H}^{(2)} \otimes (R_{H}^{(1)} \vartriangleright
a)
\end{equation}
and $A$ endowed with this coaction becomes (see \cite{Maj2, Maj3})
a (left) $H$-comodule coalgebra or equivalently a coalgebra in
${}^{H}\mathcal{M}$ (meaning that $\Delta_{A}$ and
$\varepsilon_{A}$ are (left) $H$-comodule morphisms, see \cite{Mon}). \\
We recall here (see: \cite{Maj2, Maj3}) that when $H$ is a Hopf
algebra and $A$ is a (left) $H$-comodule coalgebra with the (left)
$H$-coaction given by: $\rho(a) = \sum a^{(1)} \otimes a^{(0)}$ ,
one may form the cross coproduct coalgebra $A \rtimes H$, which as
a k-vector space is $A \otimes H$ (i.e. we write: $a \rtimes h
\equiv a \otimes h$ for every $a \in A$, $h \in H$), with
comultiplication given by:
\begin{equation} \label{eq:crosscoprcoalg}
\Delta(a \otimes h) = \sum a_{1} \otimes a_{2}^{ \ (1)} \ h_{1}
\otimes a_{2}^{ \ (0)} \otimes h_{2}
\end{equation}
and counit: $\varepsilon(a \otimes h) = \varepsilon_{A}(a)
\varepsilon_{H}(h)$. (In the above: $\Delta_{A}(a) = \sum a_{1}
\otimes a_{2}$ and we use in the elements of $A$ upper indices
included in parenthesis to denote the components of the coaction
according to the Sweedler notation, with the convention that
$a^{(i)} \in H$
for $i \neq 0$). \\
Now we proceed by applying the above described construction of the
cross coproduct coalgebra $A \rtimes H$ , with the special form of
the (left) coaction given by eq. \eqref{eq:act-coact}. Replacing
thus eq. \eqref{eq:act-coact} into eq. \eqref{eq:crosscoprcoalg}
we get for the special case of the quasitriangular Hopf algebra H
the cross coproduct comultiplication:
\begin{equation} \label{eq:crosscoprcoalgR}
\Delta(a \otimes h) = \sum a_{1} \otimes R_{H}^{(2)}h_{1} \otimes
(R_{H}^{(1)} \vartriangleright a_{2}) \otimes h_{2}
\end{equation}
Finally we can show that the cross product algebra (with
multiplication given by \eqref{eq:crosspralg}) and the cross
coproduct coalgebra (with comultiplication given by
\eqref{eq:crosscoprcoalgR}) fit together and form a bialgebra
(see: \cite{Maj2, Maj3, Mo, Ra}). This bialgebra, furnished with
an antipode:
\begin{equation} \label{antipodecrosspr}
S(a \otimes h) = (S_{H}(h_{2}))u(R^{(1)} \vartriangleright
S_{A}(a)) \otimes S(R^{(2)}h_{1})
\end{equation}
where $u = \sum S_{H}(R^{(2)})R^{(1)}$, and $S_{A}$ the (braided)
antipode of $A$, becomes (see \cite{Maj2}) an ordinary Hopf
algebra. This is the smash product Hopf algebra denoted $A \star
H$.  \\
Apart from the above described construction, it is worth
mentioning two more important points proved in \cite{Maj1}: First,
it is shown that if $H$ is triangular and $A$ is quasitriangular
in the category ${}_{H}\mathcal{M}$, then $A \star H$ is
(ordinarily) quasitriangular. Second, it is shown that the
category of the braided modules of $A$ ($A$-modules in
${}_{H}\mathcal{M}$) is equivalent to the category of the
(ordinary) modules of $A \star H$.

   \subsubsection{An example of Bosonisation}

 In the special case that $A$ is some super-Hopf
algebra, then: $H = \mathbb{CZ}_{2}$, equipped with it's
non-trivial quasitriangular structure, formerly mentioned. In this
case, the technique simplifies and the ordinary Hopf algebra
produced is the smash product Hopf algebra $A \star
\mathbb{CZ}_{2}$. The grading in $A$ is induced by the
$\mathbb{CZ}_{2}$-action on $A$:
\begin{equation} \label{eq:cz2action}
g \vartriangleright a = (-1)^{|a|}a
\end{equation}
for $a$ homogeneous in $A$. Utilizing the non-trivial $R$-matrix
$R_{g}$ and using eq. \eqref{eq:nontrivRmatrcz2} and eq.
\eqref{eq:act-coact} we can readily deduce the form of the induced
$\mathbb{CZ}_{2}$-coaction on $A$:
\begin{equation} \label{eq:cz2coaction}
\rho(a) = \left\{ \begin{array}{ccc}
    1 \otimes a & , & a: \textrm{even} \\
    g \otimes a & , & a: \textrm{odd} \\
\end{array} \right.
\end{equation}
Let us note here that instead of invoking the non-trivial
quasitriangular structure $R_{g}$ we could alternatively extract
the (left) coaction \eqref{eq:cz2coaction} utilizing the
self-duality of the $\mathbb{CZ}_{2}$ Hopf algebra: For any
abelian group $\mathbb{G}$ a (left) action of $\mathbb{CG}$
coincides with a (right) action of $\mathbb{CG}$. On the other
hand, for any finite group, a (right) action of $\mathbb{CG}$ is
the same thing as a (left) coaction of the dual Hopf algebra
$(\mathbb{CG})^{*}$. Since $\mathbb{CZ}_{2}$ is both finite and
abelian and hence self-dual in the sense that: $\mathbb{CZ}_{2}
\cong (\mathbb{CZ}_{2})^{*}$ as Hopf algebras, it is immediate to
see that the (left) action \eqref{eq:cz2action} and the (left)
coaction \eqref{eq:cz2coaction} are virtually the same thing.

 The above mentioned action and coaction enable us to form the
cross product algebra and the cross coproduct coalgebra according
to the preceding discussion which finally form the smash product
Hopf algebra $A \star \mathbb{CZ}_{2}$.  The grading of $A$, is
``absorbed" in $A \star \mathbb{CZ}_{2}$, and becomes an inner
automorphism:
$$
gag = (-1)^{|a|}a
$$
where we have identified: $a \star 1 \equiv a$ and $1 \star g
\equiv g$ in $A \star \mathbb{CZ}_{2}$ and $a$ homogeneous element
in $A$. This inner automorphism is exactly the adjoint action of
$g$ on $A \star \mathbb{CZ}_{2}$ (as an ordinary Hopf algebra).
The following proposition is proved -as an example of the
bosonisation technique- in \cite{Maj2}:
\begin{proposition} \label{bosonisat}
Corresponding to every super-Hopf algebra $A$ there is an ordinary
Hopf algebra $A \star \mathbb{CZ}_{2}$, its bosonisation,
consisting of $A$ extended by adjoining an element $g$ with
relations, coproduct, counit and antipode:
\begin{equation} \label{eq:HopfPBg}
\begin{array}{cccc}
  g^{2} = 1 & ga = (-1)^{|a|}ag & \Delta(g) = g \otimes g & \Delta(a) = \sum a_{1}g^{|a_{2}|} \otimes a_{2} \\
                                     \\
  S(g) = g & S(a) = g^{-|a|}\underline{S}(a) & \varepsilon(g) = 1 & \varepsilon(a) = \underline{\varepsilon}(a) \\
\end{array}
\end{equation}
where $\underline{S}$ and $\underline{\varepsilon}$ denote the
original maps of the super-Hopf algebra $A$. \\
In the case that $A$ is super-quasitriangular via the $R$-matrix $
\ \underline{R} = \sum \underline{R}^{(1)} \otimes
\underline{R}^{(2)} \ $, then the bosonised Hopf algebra $A \star
\mathbb{CZ}_{2}$ is quasitriangular (in the ordinary sense) via
the $R$-matrix: $ \ R = R_{g} \sum
\underline{R}^{(1)}g^{|\underline{R}^{(2)}|} \otimes
\underline{R}^{(2)} \ $.
 Moreover, the representations of the bosonised Hopf algebra $A
\star \mathbb{CZ}_{2}$ are precisely the super-representations of
the original superalgebra $A$.
\end{proposition}
The application of the above proposition in the case of the
parabosonic algebra $P_{B}$ is straightforward: we immediately get
it's bosonised form $P_{B(g)}$ which by definition is:
$
P_{B(g)} \equiv P_{B} \star \mathbb{CZ}_{2}
$
Utilizing equations \eqref{eq:HopfPB} which describe the
super-Hopf algebraic structure of the parabosonic algebra $P_{B}$,
and replacing them into equations \eqref{eq:HopfPBg} which
describe the ordinary Hopf algebra structure of the bosonised
superalgebra, we immediately get the explicit form of the
(ordinary) Hopf algebra structure of $P_{B(g)} \equiv P_{B} \star
\mathbb{CZ}_{2}$ which reads:
\begin{equation} \label{eq:HopfPBgexpl}
\begin{array}{cccc}
  \Delta(B_{i}^{\pm}) = B_{i}^{\pm} \otimes 1 + g \otimes B_{i}^{\pm} & \Delta(g) = g \otimes g
  & \varepsilon(B_{i}^{\pm}) = 0 & \varepsilon(g) = 1  \\
          \\
  S(B_{i}^{\pm}) = B_{i}^{\pm}g = -gB_{i}^{\pm} & S(g) = g & g^{2} = 1 & \{g,B_{i}^{\pm}\} = 0  \\
\end{array}
\end{equation}
where we have again identified $b_{i}^{\pm} \star 1 \equiv
b_{i}^{\pm}$ and $1 \star g \equiv g$ in $P_{B} \star
\mathbb{CZ}_{2}$. Finally, we can easily check that since
$\mathbb{CZ}_{2}$ is triangular (via $R_{g}$) and $P_{B}$ is
super-quasitriangular (trivially since it is super-cocommutative)
 it is an immediate consequence of the above proposition that
 $P_{B(g)}$ is quasitriangular (in the ordinary sense) via the
 $R$-matrix $R_{g}$.

\subsubsection{An alternarive approach}

Let us describe now a slightly different construction (see also:
\cite{DaKa, KaDa1, KaDa2}), which achieves the same object: the
determination of an ordinary Hopf structure for the parabosonic
algebra $P_{B}$.
\begin{proposition} \label{altern}
Corresponding to the super-Hopf algebra $P_{B}$ there is an
ordinary Hopf algebra $P_{B(K^{\pm})}$, consisting of $P_{B}$
extended by adjoining two elements $K^{+}$, $K^{-}$ with
relations, coproduct, counit and antipode:
\begin{equation} \label{eq:HopfPBK}
\begin{array}{cc}
  \Delta(B_{i}^{\pm}) = B_{i}^{\pm} \otimes 1 + K^{\pm} \otimes B_{i}^{\pm} & \Delta(K^{\pm}) = K^{\pm} \otimes K^{\pm} \\
       \\
  \varepsilon(B_{i}^{\pm}) = 0 & \varepsilon(K^{\pm}) = 1 \\
                     \\
  S(B_{i}^{\pm}) = B_{i}^{\pm}K^{\mp} & S(K^{\pm}) = K^{\mp} \\
                      \\
  K^{+}K^{-} = K^{-}K^{+} = 1 & \{K^{+},B_{i}^{\pm}\} = 0 = \{K^{-},B_{i}^{\pm}\} \\
\end{array}
\end{equation}
\end{proposition}

\begin{proof}
Consider the vector space $\mathbb{C}\langle X_{i}^{+}, X_{j}^{-},
K^{\pm} \rangle$ freely generated by the elements $X_{i}^{+},
X_{j}^{-}, K^{+}, K^{-}$.  Denote $T(X_{i}^{+}, X_{j}^{-},
K^{\pm})$ its tensor algebra. In the tensor algebra we denote
$I_{BK}$ the ideal generated by all the elements of the forms
\eqref{eq:pbdef} together with: $\ K^{+}K^{-}-1 \ $, $\
K^{-}K^{+}-1 \ $, $\ \{K^{+},X_{i}^{\pm}\} \ $, $\
\{K^{-},X_{i}^{\pm}\} \ $. We define:
$$
P_{B(K^{\pm})} = T(X_{i}^{+}, X_{j}^{-}, K^{\pm}) / I_{BK} \
$$
We denote by $B_{i}^{\pm}, K^{\pm}$ where $i = 1, 2, \ldots \ $
the images of the generators $X_{i}^{\pm}, K^{\pm}$, $ \ i = 1, 2,
\ldots \ $ of the tensor algebra, under the canonical projection.
These are a set of generators of $P_{B(K^{\pm})}$. Consider the
linear map $ \Delta : \mathbb{C}\langle X_{i}^{+}, X_{j}^{-},
K^{\pm} \rangle \rightarrow P_{B(K^{\pm})} \otimes P_{B(K^{\pm})}
$ determined by: determined by
$$
\begin{array}{c}
\Delta(X_{i}^{\pm}) = B_{i}^{\pm}
\otimes 1 + K^{\pm} \otimes B_{i}^{\pm} \\
    \\
 \Delta(K^{\pm}) = K^{\pm} \otimes K^{\pm}  \\
\end{array}
$$
By the universality property of the tensor algebra, this map
extends to an algebra homomorphism: $ \Delta: T(X_{i}^{+},
X_{j}^{-}, K^{\pm}) \rightarrow P_{B(K^{\pm})} \otimes
P_{B(K^{\pm})} $. We emphasize that the usual tensor product
algebra $P_{B(K^{\pm})} \otimes P_{B(K^{\pm})}$ is now considered,
with multiplication $(a \otimes b)(c \otimes d) = ac \otimes bd$
for any $a,b,c,d \in P_{B(K^{\pm})}$.
 Now we can
trivially verify that
\begin{equation} \label{eq:DKb}
 \Delta(\{K^{\pm},X_{i}^{\pm}\})
= \Delta(K^{+}K^{-} -1) = \Delta(K^{-}K^{+}-1) = 0
\end{equation}
We also compute:
\begin{equation} \label{eq:Db}
\Delta(\big[ \{ X_{i}^{\xi},  X_{j}^{\eta}\}, X_{k}^{\epsilon}
\big] -
 (\epsilon - \eta)\delta_{jk}X_{i}^{\xi} - (\epsilon -
 \xi)\delta_{ik}X_{j}^{\eta}) = 0
\end{equation}
Relations \eqref{eq:DKb}, and \eqref{eq:Db}, mean that $I_{BK}
\subseteq ker \Delta$ which in turn implies that $\Delta$ is
uniquely extended as an algebra homomorphism from $
P_{B(K^{\pm})}$ to  the usual tensor product algebra
$P_{B(K^{\pm})} \otimes P_{B(K^{\pm})}$, with the values on the
generators determined by \eqref{eq:HopfPBK}. \\
 Following the same
procedure we construct an algebra homomorphism $\varepsilon:
P_{B(K^{\pm})} \rightarrow \mathbb{C}$ and an algebra
antihomomorphism $S: P_{B(K^{\pm})} \rightarrow P_{B(K^{\pm})}$
which are completely determined by their values on the generators
of $P_{B(K^{\pm})}$ (i.e.: the basis elements of
$\mathbb{C}\langle X_{i}^{+}, X_{j}^{-}, K^{\pm} \rangle)$. Note
that in the case of the antipode we start by defining a linear map
$S$ from $\mathbb{C}\langle X_{i}^{+}, X_{j}^{-}, K^{\pm} \rangle$
to the opposite algebra $P_{B(K^{\pm})}^{op}$, with values
determined by: $S(X_{i}^{\pm}) = B_{i}^{\pm}K^{\mp}$ and
$S(K^{\pm}) = K^{\mp} \ $. Following the above described procedure
we end up with an algebra anti-homomorphism:
$S: P_{B(K^{\pm})} \rightarrow P_{B(K^{\pm})}$.  \\
Now it is sufficient to verify the rest of the Hopf algebra axioms
(i.e.: coassociativity of $\Delta$, counity property for
$\varepsilon$, and the compatibility condition which ensures us
that $S$ is an antipode) on the generators of $P_{B(K^{\pm})}$.
This can be done with straightforward computations (see
\cite{DaKa}).
\end{proof}
Let us notice here, that the initiation for the above mentioned
construction lies in the case of the finite degrees of freedom: If
we consider the parabosonic algebra in $2n$ generators
($n$-paraboson algebra) and denote it $P_{B}^{(n)}$, it is
possible to construct explicit realizations of the elements
$K^{+}$ and $K^{-}$ in terms of formal power series, such that the
relations specified in \eqref{eq:HopfPBK} hold. The construction
is briefly (see also \cite{DaKa}) as follows: We define
$$
\mathcal{N} = \sum_{i=1}^{n}N_{ii} =
\frac{1}{2}\sum_{i=1}^{n}\{B_{i}^{+},B_{i}^{-}\}
$$
We inductively prove:
\begin{equation} \label{eq:Casimcom}
[\mathcal{N}^{m}, B_{i}^{+}] = B_{i}^{+}((\mathcal{N} + 1)^{m} -
\mathcal{N}^{m})
\end{equation}
We now introduce the following elements:
$$
\begin{array}{ccccc}
  K^{+} = \exp(i \pi \mathcal{N}) & & & & K^{-} = \exp(-i \pi \mathcal{N}) \\
\end{array}
$$
Utilizing the above power series expressions and equation
\eqref{eq:Casimcom} we get
\begin{equation} \label{eq:Kb}
\begin{array}{lr}
 \{K^{+},B_{i}^{\pm}\} = 0 &  \{K^{-},B_{i}^{\pm}\} = 0 \\
\end{array}
\end{equation}
A direct application of the Baker-Campbell-Hausdorff formula leads
also to:
\begin{equation} \label{eq:KK}
K^{+}K^{-} = K^{-}K^{+} = 1
\end{equation}
Finally let us make a few comments on the above mentioned
constructions.

From the point of view of the structure, an obvious question
arises: While $P_{B(g)}$ is a quasitriangular Hopf algebra through
the $R$-matrix: $R_{g}$ given in eq. \eqref{eq:nontrivRmatrcz2},
there is yet no suitable $R$-matrix for the Hopf algebra
$P_{B(K^{\pm})}$. Thus the question of the quasitriangular
structure of $P_{B(K^{\pm})}$ is open.

On the other hand, regarding representations, we have already
noted that the super representations of $P_{B}$ (
$\mathbb{Z}_{2}$-graded modules of $P_{B}$ or equivalently:
$P_{B}$-modules in
 ${}_{\mathbb{CZ}_{2}}\mathcal{M}$ ) are in $``1-1"$ correspodence
 with the (ordinary) representations of $P_{B(g)}$. Although we do
 not have such a strong result for the representations of
 $P_{B(K^{\pm})}$, the preceding construction in the case of
 finite degrees of freedom enables us to uniquely extend the Fock-like
 representations of $P_{B}^{(n)}$ to representations of
 $P_{B(K^{\pm})}^{(n)}$. Since the Fock-like representations of
 $P_{B}$ are unique up to unitary equivalence (see the proof in
 \cite{GreeMe} or \cite{OhKa}), this is a point which deserves to be discussed
 analytically in a forthcoming work.

\textbf{Acknowledgements:} This paper is part of a
 project supported by ``Pythagoras II", contract number 80897.


\begin{thebibliography}{99}

\bibitem{Andru} N. Andruskiewitsch, P. Etingof, S. Gelaki, Michigan Math. J.,
 v.\textbf{49} , (2001),  p.277

\bibitem{DaKa} C. Daskaloyannis, K. Kanakoglou, I. Tsohantjis, J. Math. Phys., v.\textbf{41}, 2,
(2000), p.652

\bibitem{Dri} V. G. Drinfeld, ``Quantum Groups", Proc. Int. Cong.
Math., Berkeley, (1986), p. 789-820

\bibitem{Ehrenf} P. Ehrenfest, Zeits. f. Physik, v.\textbf{45},
(1927), p.455

\bibitem{Fi} D. Fischman, J. Algebra, v.\textbf{157}, (1993), p.331

\bibitem{FiMon} D. Fischman, S. Montgomery, J.
Algebra, v.\textbf{168}, (1994), p.594

\bibitem{Pal} A. Ch. Ganchev, T.D. Palev, J. Math. Phys., v.\textbf{21}, 4,
(1980), p.797.

\bibitem{GreeMe} O.W. Greenberg, A.M.L. Messiah, Phys.
Rev., v.\textbf{138}, 5B, (1965), p.1155

\bibitem{Green} H.S. Green, Phys. Rev., v.\textbf{90}, 2, (1953), p.270

\bibitem{KaDa1} K. Kanakoglou, C. Daskaloyannis, Proceedings of the 6th
Panhellenic conference in Algebra and Number Theory, Thessaloniki,
GREECE, 10-12 June, 2006, arXiv:hep-th/0701077.

\bibitem{KaDa2} K. Kanakoglou, C. Daskaloyannis, Talk
presented to ``AGMF: Algebra, Geometry, and Mathematical Physics",
Baltic-Nordic Workshop: Lund, SWEDEN, 12-14 October, 2006,
arXiv:math-ph/0701023.

\bibitem{Ko} B. Konstant, Lect. Notes in Math., v.\textbf{570},
Springer, (1977), p.177-306.

\bibitem{Maj1} S. Majid, J. Alg., v.\textbf{163}, (1994), p.165

\bibitem{Maj2} S. Majid, ``Foundations of Quantum Group Theory",
Cambridge University Press, 1995.

\bibitem{Maj3} S. Majid, ``A quantum groups primer", London
Mathematical Society, Lecture Note Series, 292, Cambridge
University Press, 2002.

\bibitem{MiMo} J. Milnor, J. Moore, Ann. of Math., v.\textbf{81}, (1965), p.211-264.

\bibitem{Mo} R.K. Molnar, J. Alg., v.\textbf{47}, (1977), p.29

\bibitem{Mon} S. Montgomery, ``Hopf
algebras and their actions on rings", CBMS, Regional Conference
Series in Mathematics, 82, AMS-NSF, 1993.

\bibitem{OhKa} Y. Ohnuki, S. Kamefuchi, ``Quantum field theory and
parastatistics", University of Tokyo press, Tokyo, Springer, 1982.

\bibitem{Pal1} T.D. Palev, Reports on Mathematical Physics,
v.\textbf{31}, 3, (1992), p.241-262

\bibitem{Ra} D.E. Radford, J. Alg., v.\textbf{92}, (1985), p.322

\bibitem{Scheu} M. \ Scheunert, Lect. Not. Math., v.\textbf{716}, Springer, (1978), p.1-270.

\bibitem{Stee} N.E. Steenrod, Enseign. Math. II., ser.\textbf{7},
(1961), p.153-178.

\bibitem{VanDerWar} B.L. van der Waerden, ed., ``Sources of Quantum
mechanics", North-Holland, Amsterdam, 1967. Reprinted by Dover,
Classics of Science, vol. 5, 1968.

\bibitem{Wi} E.P. Wigner, Physical Review, v.\textbf{77}, 5, (1950), p.711-712.

\end{thebibliography}
\end{document}